\documentclass[journal,comsoc]{IEEEtran}

\IEEEoverridecommandlockouts

\usepackage{todonotes}
\usepackage{amsmath}
\allowdisplaybreaks
\usepackage{amssymb,amsfonts}
\usepackage{mathtools}
\usepackage{algorithm}
\usepackage{algpseudocode}
\usepackage{array}
\usepackage{graphicx}
\usepackage{xspace}
\usepackage[nopostdot,acronym,shortcuts,nonumberlist]{glossaries}
\usepackage{booktabs}
\usepackage[hyphens]{url}
\usepackage{enumitem}
\usepackage{cite}
\usepackage{comment}
\usepackage[bookmarks=false,hidelinks]{hyperref}
\usepackage{subcaption}
\usepackage{setspace}
\usepackage[table, dvipsnames]{xcolor}
\usepackage{makecell}

\usepackage{bbm}
\usepackage{dsfont}
\usepackage{amsthm}
\usepackage{optidef}
\PassOptionsToPackage{bookmarks=false}{hyperref}
\theoremstyle{plain}
\newtheorem{lemma}{Lemma}
\newtheorem{theorem}{Theorem}
\theoremstyle{definition}
\newtheorem{definition}{Definition}

\newcommand{\rnats}{\mathbb{N}_0}
\newcommand{\pnats}{\mathbb{N}^+}
\newcommand{\reals}{\mathbb{R}}
\newcommand{\indicate}{\mathbb{I}}
\DeclareMathOperator*{\maximum}{\mathrm{max}} 
\DeclareMathOperator*{\limit}{\mathrm{lim}} 
 
\DeclareMathOperator*{\argmax}{arg\,max}
\newcommand{\opti}[1]{#1^\ast}
\DeclareMathOperator*{\expect}{\mathbb{E}}
\newcommand{\cexpect}[2]{\expect_{#1\sim #2}}
\newcommand{\states}{\mathcal{S}}
\newcommand{\actions}{\mathcal{A}}
\newcommand{\V}{V^\ast}
\newcommand{\Q}{Q^\ast}
\newcommand{\aoi}[1]{\mathrm{AoI}_{#1}}
\newcommand{\apaoi}[1]{\mathrm{AoI}^\mathrm{AP}_{#1}}
\newcommand{\maxaoi}{\Lambda}
\newcommand{\pwk}{p^\mathrm{w}}
\newcommand{\ptx}{p^\mathrm{Tx}}

\newcommand{\threshold}{\kappa}
\newcommand{\param}{\theta}

\newacronym{ud}{UD}{User Device}
\newacronym{es}{ES}{Edge Server}
\newacronym{pu}{PU}{Processing Unit}
\newacronym{gpu}{GPU}{Graphical Processing Unit}
\newacronym{flop}{FLOP}{Floating Point Operations}
\newacronym{flops}{FLOP/s}{Floating Point Operations per Second}
\newacronym{sut}{SUT}{Semantic Unit}
\newacronym{qoe}{QoE}{Quality-of-Experience }

\newacronym{bs}{BS}{Base Station}
\newacronym{snr}{SNR}{signal-to-noise ratio}
\newacronym{drl}{DRL}{Deep Reinforcement Learning}
\newacronym{marl}{MARL}{Multi-Agent Reinforcement Learning}
\newacronym{posg}{POSG}{Partially Observable Stochastic Game}
\newacronym{dqn}{DQN}{Deep Q-Network}

\newacronym{mec}{MEC}{Mobile Edge Computing}
\newacronym{ppo}{PPO}{Proximal Policy Optimisation}
\newacronym{mappo}{MAPPO}{Multi-Agent Proximal Policy Optimisation}

\newacronym{6G}{6G}{sixth-generation}

\newacronym{rl}{RL}{Reinforcement Learning}
\newacronym{iot}{IoT}{Internet-of-Things}
\newacronym{id}{device}{Internet-of-Things device}
\def\id{device\xspace}
\def\ids{devices\xspace}

\newacronym{aoi}{AoI}{Age of Information}
\newacronym{ap}{AP}{Access Point}

\newacronym{mdp}{MDP}{Markov Decision Process}

\newacronym{ann}{ANN}{Artificial Neural Network}
\newacronym{snn}{SNN}{Spiking Neural Network}
\newacronym{ifn}{IF-neuron}{Integrate and Fire Neuron}

\begin{document}

\begin{titlepage}
    \centering
    \vspace*{\fill}

\centering
\textbf{Preprint Notice}\\
\vspace*{1em}
\fbox{
    \begin{minipage}{0.8\textwidth}
        \centering
        \vspace{0.8em}
        This work has been submitted to the IEEE for possible publication.
        Copyright may be transferred without notice, after which this
        version may no longer be accessible.
        \vspace{0.8em}
    \end{minipage}
}
\vspace*{\fill}
\end{titlepage}

\bstctlcite{IEEEexample:BSTcontrol}

\title{Threshold-Based Spiking Neural Networks for Event-Driven Status Update  Systems} 

\author{Marco Fries, Andrea Ortiz\\
\IEEEauthorblockA{
Institute of Telecommunications, Vienna University of Technology, Austria,\\\text{marco.fries@tuwien.ac.at, andrea.ortiz@tuwien.ac.at}}
\thanks{This work is funded by the Vienna Science and Technology Fund (WWTF) [Grant ID: 10.47379/VRG23002] and by the IoT-ZERO project which received funding from the Smart Networks and Services Joint Undertaking (SNS JU) under the European
Union’s Horizon Europe research and innovation program under Grant Agreement
No. 101292662.}}

\maketitle

\begin{abstract}
Event-driven sensing supports energy-efficient \gls{iot} devices by activating communication only when relevant events occur. In such systems, transmission decisions are governed by the monitored process rather than predefined schedules.  Consequently, jointly optimising information freshness and energy consumption is challenging because transmission decisions are restricted to randomly occurring events. To address this challenge, we investigate an event-driven status update system in which wake-up events follow the dynamics of the monitored process. The problem of determining whether to transmit the sensing data or not is cast as a \gls{mdp} that jointly minimises the \gls{aoi} and transmission energy. We prove the existence of an optimal threshold policy, thereby obtaining an interpretable characterisation of the optimal transmission strategy. Motivated by this result, we propose a lightweight \gls{rl} approach based on \glspl{snn} whose architecture explicitly represents threshold policies. The resulting policy representation has constant complexity with respect to the maximum \gls{aoi} and enables a more energy-efficient implementation than a comparable \gls{ann}. Numerical results demonstrate that the proposed \gls{snn} reliably learns optimal thresholds across different operating regimes.
\end{abstract}
\section{Introduction}
\label{sec:intro}
\IEEEPARstart{T}{imely} information delivery is a fundamental requirement in many monitoring and sensing applications, particularly when sensing devices can only communicate intermittently. Such systems can be modelled as status update systems, where sensing devices observe a physical process and communicate updates to a remote destination. A key performance metric in these systems is information freshness, commonly quantified by the \gls{aoi} \cite{Yates2021}.

In many practical deployments, sensing devices are battery-powered or rely on energy harvesting. Consequently, energy-efficient operation is essential for maintaining long-term connectivity and ensuring sustainable information updates. A common approach is to place devices into deep-sleep states, during which energy consumption is minimised. Existing models of status update systems typically assume that transmission opportunities and sleep durations are determined by the transmitter, a scheduler, or a predefined control policy.
However, event-driven operation is increasingly employed in sensing applications  \cite{Kolios2016,SamurAI2023}. In such systems, communication decisions are triggered by the evolution of the monitored process rather than by externally imposed schedules. As a result, the communication process becomes intrinsically coupled with the dynamics of the monitored system.

While prior work has extensively investigated \gls{aoi}-optimal scheduling and energy-efficient status updating \cite{Zhou,ceran2021,Cao2023,DeSombre2023,Dongare,DeSombre2026}, the structure of optimal policies for process-driven sleep behaviour remains largely unexplored. In particular, it is unclear whether such systems admit simple optimal policies that can be efficiently implemented on resource-constrained devices.
A comparison of related work and the positioning of the proposed approach is provided in Table~\ref{tab:sota}.
In this letter, we investigate an event-driven status update system in which the sleep behaviour of the \gls{iot} \ids is governed by the evolution of the monitored process. Our main theoretical result establishes the existence of an optimal threshold policy. This characterisation yields an interpretable solution approach and provides fundamental insights into the freshness-energy trade-off of event-driven status update systems. Building on this result, we develop a lightweight learning-based \gls{snn} implementation, enabling efficient policy optimisation with a complexity that does not scale with the maximum \gls{aoi}.
The main contributions of this letter are summarised as follows:
\begin{itemize}
\item We introduce a mathematical framework for modelling and analysing event-driven status update systems with process-dependent sleep behaviour.
\item We prove the existence of an optimal threshold policy for the considered event-driven model.
\item We develop a lightweight policy-gradient \gls{snn} architecture that approximates an optimal threshold policy while maintaining a complexity that does not scale with the maximum \gls{aoi}.
\item We validate the theoretical findings through numerical simulations and demonstrate the effectiveness and energy efficiency of the proposed approach.
\end{itemize}

The remainder of this letter is organised as follows. Section~\ref{sec:systModel} presents the system model. Section~\ref{sec:problem} formulates the optimisation problem. Section~\ref{sec:solution} describes the proposed solution approach.Section~\ref{sec:energy} compares the inference energy consumption of the proposed solution and an \gls{ann} benchmark. Numerical results are provided in Section~\ref{sec:results}, followed by conclusions in Section~\ref{sec:conclusion}.

\begin{table}[t]
\begin{center}
    \scriptsize
    \renewcommand{\arraystretch}{0.6}
    \begin{tabular}{|c|*{7}{>{\centering\arraybackslash}m{0.35cm}|}}
        \hline
        & \makecell*{\hspace{-0.1cm}\cite{Zhou}} & \makecell*{\hspace{-0.1cm}\cite{ceran2021}} & \makecell*{\hspace{-0.1cm}\cite{Cao2023}} & \makecell*{\hspace{-0.1cm}\cite{DeSombre2023}} & \makecell*{\hspace{-0.1cm}\cite{Dongare}} & \makecell*{\hspace{-0.1cm}\cite{DeSombre2026}} & \cellcolor[gray]{0.9} \makecell*{\vspace{-0.15cm}\hspace{-0.15cm} Ours} \\
        \hline
        \hline
        \makecell*{Energy Awareness} 
        & \checkmark & \checkmark & \checkmark & \checkmark & \checkmark & \checkmark & \cellcolor[gray]{0.9}\checkmark \\
        \hline
        \makecell*{Random sensing} 
        & \checkmark & \checkmark & & \checkmark  & \checkmark & \checkmark & \cellcolor[gray]{0.9}\checkmark \\
        \hline
        \makecell*{Event-based} 
        & & & &  & & & \cellcolor[gray]{0.9}\checkmark \\
        \hline
        \makecell*{Unknown dynamics} 
        & & & \checkmark & \checkmark & \checkmark & \checkmark & \cellcolor[gray]{0.9}\checkmark \\
        \hline
    \end{tabular}
    \caption{Summary of the related work.}
    \label{tab:sota}
\end{center}
\end{table}

\textit{Notation:} The indicator function $\mathbb{I}[\phi]$ takes the value $1$ if the Boolean statement $\phi$ holds and $0$ otherwise. For any set $\mathcal{X}$, $\Delta\mathcal{X}$ denotes the simplex of probability distributions over $\mathcal{X}$. By $X\sim\mathbb{P}(\cdot\mid\boldsymbol{\theta})$, we indicate that $X$ is a random variable distributed according to $\mathbb{P}(\cdot\mid\boldsymbol{\theta})$. Expectations with respect to a distribution are denoted by $\mathbb{E}_{X\sim\mathbb{P}(\cdot\mid\boldsymbol{\theta})}[\cdot]$.

\section{System Model}
\label{sec:systModel}
We consider a status-update system consisting of a battery-powered \gls{iot} \id and an \gls{ap}. The \id monitors a physical process through an event-based sensor, e.g., a neuromorphic camera. Status updates can be transmitted from the \id to the \gls{ap} via a wireless point-to-point communication channel. The \gls{ap} requires fresh data, and the goal of the \id is to provide fresh updates to the \gls{ap} while minimising its own energy consumption in order to remain operational for an extended period of time.

The \id operates in two modes, namely \textit{deep-sleep} mode and \textit{awake} mode.
While in deep-sleep mode, the \id's energy consumption is minimal. In this state, the event-based sensor can trigger a \textit{wake-up} event at the \id. In that case, the \id transitions from deep-sleep mode to awake mode, where it receives a status update, decides whether or not to transmit it to the \gls{ap}, and then transitions back to deep-sleep mode.
Time is divided into discrete time steps $t\in\mathbb{N}$ of a fixed duration $\tau$, measured in seconds. 
The wake-up events are modelled as a Bernoulli process, meaning that 
at any time step $t$, a wake-up event occurs with fixed probability $\pwk\in(0,1]$. Let $\mathcal{W}:=\{w_1,\dots\}\subseteq\rnats$ denote the set of time steps at which a wake-up event occurs.
The channel between the \id and the \gls{ap} is modelled as a binary erasure channel with a success probability $\ptx\in(0,1]$. This means a transmission succeeds with probability $\ptx$ and fails with probability $1-\ptx$.

We use \gls{aoi} as a metric to measure data-freshness. The \gls{aoi} at $t\in\mathbb{N}^+$ is inductively defined by
\begin{equation}
    \aoi{t+1}:=\begin{cases}
        0&\text{on successful data update},\\
        \aoi{t}+1&\text{otherwise,}
    \end{cases}\label{eq:aoibase}
\end{equation}
and $\aoi{0}:=0$.
We assume that transmissions from \id to \gls{ap} are performed within one time step and that there is an upper bound $\maxaoi\in\mathbb{N}$, after which an increase of \gls{aoi} at the \gls{ap} is indifferent for the application at hand. Based on \eqref{eq:aoibase}, the \gls{aoi} of the \gls{ap} at the $(t+1)$-th wake-up event $w_{t+1}$, $t\in\mathbb{N}_0$, evolves according to
\begin{equation}
    \apaoi{w_{t+1}}:=\begin{cases}
        \delta_t&\text{on successful data}\\
        &\text{transmission at }w_t,\\
        \min\{{\apaoi{w_t}+\delta_t,\maxaoi}\}&\text{otherwise,}
    \end{cases}
\end{equation}
where $\delta_t$ is the deep-sleep duration between $w_t$ and $w_{t+1}$. Note that $\delta_t$ follows a geometric distribution with parameter $\pwk$. We assume that there is an instantaneous, perfect feedback channel from the \gls{ap} to the \id, enabling the \id to keep track of $\apaoi{t}$ without delay. Since the transmission decisions are restricted to wake-up events $w\in\mathcal{W}$, the \gls{aoi} of the \id is $0$ when such a decision is made. Thus, we omit the \gls{aoi} at the \id in our model.

We assume that the amount of energy available to the \id is limited by a battery with a capacity of $E^\mathrm{Battery}\in\reals^+$ and that the \id consumes a fixed amount of energy, $0<E^{\mathrm{Tx}}\leq E^\mathrm{Battery}$, per transmission.
To ensure long-term operability, our goal is to optimise the trade-off between the \id's energy consumption and the data-freshness at the \gls{ap}.

\begin{figure}
    \centering
    \includegraphics[width=0.95\linewidth]{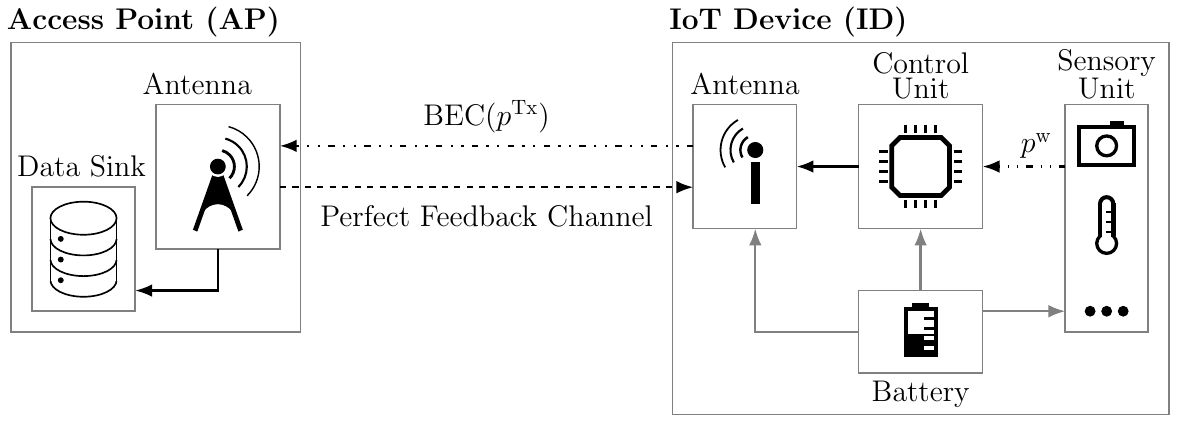}
    \vspace{-0.1cm}
    \caption{System model}
    \label{fig:sysmodel}
    \vspace{-0.3cm}
\end{figure}
\section{Problem formulation}
\label{sec:problem}
In this section, we formulate the transmission decision-making process of the device as a \gls{mdp}. At each wake-up event $w\in\mathcal{W}$, the \id observes the $\apaoi{w}$ and decides whether or not to transmit its current data to the \gls{ap}. The goal of the \id is to jointly minimise the \gls{aoi} at the \gls{ap} and its own energy consumption. The components of the \gls{mdp} $\mathcal{M}=(\mathcal{S},\mathcal{A},F,r)$ are specified as follows:
\noindent \\
\textbf{The set of states} $\states:=\{1,\dots,\maxaoi\}$ consists of all observable values of $\apaoi{}$.\\
\textbf{The set of actions} $\actions:=\{0,1\}$ represents the transmission decisions of the \id. Here $1$ indicates that a transmission attempt is made, and $0$ indicates otherwise. \\
\textbf{The transition probability} $F(s^\prime|s,a)$ denotes the probability of transitioning to the state $s^\prime$ when taking action $a$ in state $s$. It arises from the geometric distribution governing the deep-sleep durations between wake-up events and the probability of a successful transmission. More precisely
\begin{align*}
F(s^\prime\mid s,0)&=\indicate[s<s^\prime](1-\pwk)^{s^\prime-s-1}\pwk\\
    &\phantom{=}+\indicate[s^\prime=\maxaoi](1-\pwk)^{\maxaoi-s},
\end{align*}
and
\begin{align*}
    F(s^\prime\mid s,1)&=\ptx F(s^\prime\mid0,0)\\
    &\phantom{=}+(1-\ptx)F(s^\prime\mid s,0).
\end{align*}
\textbf{The reward function} $r$ captures a trade-off between $\apaoi{}$, weighted by $c^\mathrm{AoI}\in\reals^+$, and the \id's energy consumption, weighted by $c^\mathrm{E}\in\reals^+$
\begin{align*}
    r\colon\mathcal{S}\times\mathcal{A}\to\reals;\quad
    (s,a)\mapsto -c^{\mathrm{AoI}}s - c^{\mathrm{E}}a\mathrm{E}^\mathrm{Tx}.
\end{align*}
The coefficients $c^\mathrm{AoI}$ and $c^\mathrm{E}$ quantify the relative importance assigned to information freshness and energy consumption, respectively. Accordingly, $c^\mathrm{AoI}$ is measured in $\mathrm{s}^{-1}$ and $c^\mathrm{E}$ in $\mathrm{J}^{-1}$. 
At each wake-up event of the system, the \id makes a transmission decision and receives an immediate reward signal.
A \textit{trajectory} $\psi=(s_0,a_0,\ldots,s_{H-1},a_{H-1},s_H)$ of length $H\in\pnats$ in $\mathcal{M}$ records the observed values of $\apaoi{}$ and the corresponding transmission decisions of the \id over successive wake-up events. For a trajectory $\psi$, its average return is
$R(\psi):=\frac{1}{H}\sum_{t=0}^{H-1} r(s_t,a_t)$. The probability $\mu(\psi\mid\pi)$ of observing the trajectory $\psi$ under a policy $\pi\colon\states\to\Delta\actions$ is
$$\mu(\psi\mid\pi):=\nu(s_0)\prod_{t=0}^{H-1}\pi(a_t\mid s_t)F(s_{t+1}\mid s_t,a_t),$$
where $\nu\in\Delta\states$ denotes the initial-state distribution.
We aim to find a policy maximising the expected long-term average return
\begin{equation}
    \max_{\pi\colon\states\to\Delta\actions}\hspace{0.5em}
    \limit_{H\to\infty}\hspace{0.5em}\cexpect{\psi}{\mu(\cdot\mid\pi)}[R(\psi)].\label{eq:optProblem}
\end{equation}

\section{Proposed solution}
\label{sec:solution}
\subsection{Optimality of threshold policies}
In this subsection, we first prove that the optimisation problem \eqref{eq:optProblem} admits an optimal threshold policy. Motivated by this result, we then propose a lightweight \gls{rl} approach.

\begin{definition}\label{def:pith}
 For $\threshold\in\reals$, the policy $\vartheta_\threshold\colon\states\to\Delta\actions$ defined by
\begin{equation*}
    \vartheta_\threshold(0\mid s):=\indicate[s<\threshold]\text{ and }\vartheta_\threshold(1\mid s):=\indicate[\threshold\leq s]
\end{equation*}
is called the \textit{threshold policy} with \textit{threshold} $\threshold$.
\end{definition}

\begin{theorem}\label{thm:main}
    There exists $\threshold^\ast\in\reals$ such that $\vartheta_{\threshold^\ast}$ solves \eqref{eq:optProblem}.
\end{theorem}
\begin{proof}
    Cf. Appendix \ref{sec:optpi}.
\end{proof}

While Theorem~\ref{thm:main} establishes the existence of an optimal threshold $\threshold^\ast$, determining its value requires knowledge of the underlying system dynamics, such as the channel quality and the wake-up probability \cite{Cao2023}.
Since these quantities are unknown a priori, we employ \gls{rl} and an \gls{snn}-based policy representation to approximate the optimal threshold policy through interaction with the environment.
\subsection{Preliminaries on Spiking Neural Networks}

We consider \glspl{snn} due to their potential for energy-efficient computation, particularly when implemented on specialised neuromorphic hardware.
Inspired by the information processing mechanisms observed in biological neural systems, information in \glspl{snn} is communicated through discrete, threshold-triggered \textit{spikes}.
If information is communicated only through sparse spike events, the number of computations and memory accesses can be substantially reduced compared with conventional \glspl{ann}.
To this end, each neuron in a \gls{snn} maintains an internal \textit{membrane potential} and an associated \textit{firing threshold}. The input of a neuron is encoded as \textit{input current}, which is added to its membrane potential. When the membrane potential exceeds the firing threshold, the neuron emits a spike to its successor neurons.

In this work, we exploit the inherent threshold mechanism of spiking neurons to obtain a policy representation that directly mirrors the structure of a threshold policy (cf. Definition~\ref{def:pith}). Consequently, the resulting \gls{snn} architecture admits an interpretable parametrisation and yields a compact representation whose complexity does not scale with the maximum \gls{aoi}.

\subsection{Proposed SNN-Architecture}
\label{subsec:snn}

The proposed architecture employs an \gls{snn} consisting of two \glspl{ifn} that implement a randomised threshold policy $\zeta_{\threshold,\omega}\colon\states\to\Delta\actions$. For this, the observed $\apaoi{}$ $s$ is encoded in the neuron's input currents $\hat{I}$ and $\check{I}$ defined as
\begin{equation*}
    \hat{I}(s):=s+\indicate[s=\threshold]\varepsilon;\quad
    \check{I}(s):=2\threshold-s,
\end{equation*}
where $\varepsilon>0$ is a small constant.
Each neuron has its own membrane potentials $\hat{U}$, respectively $\check{U}$, initially set to $0$. The neurons  share a common firing threshold $\threshold\in\reals$ and spike-grading factor $\omega>0$, which together are the learnable parameters of the architecture. The input current of a neuron is added to its membrane potential. If a neurons potential surpasses $\threshold$, it outputs a \textit{graded spike} of amplitude $\omega>0$; and $0$ otherwise. The resulting spiking behaviour is
\begin{equation*}
        \hat{Z}=\omega\indicate[\hat{U}>\threshold]\text{ respectively }\check{Z}=\omega\indicate[\check{U}>\threshold].
\end{equation*}
By construction
$
\hat{Z}=\omega\vartheta_\threshold(1\mid s)
$
and
$
\check{Z}=\omega\vartheta_\threshold(0\mid s),
$
meaning that the spiking behaviour of the neurons emulates the decision making of the threshold policy $\vartheta_\threshold$. 
The spiking behaviour of both neurons determines the network's output
\begin{align*}
    \tilde{I}:=\hat{Z}-\check{Z}=\omega[\vartheta_\threshold(1\mid s)-\vartheta_\threshold(0\mid s)].
\end{align*} 
This output is then decoded into the probability of making a transmission attempt
$
    \zeta_{\threshold,\omega}(1\mid s):=\sigma(\tilde{I})\in(0,1),
$
where $\sigma(x)=\frac{1}{1+exp(-x)}$.
Finally, the transmission decision $a\in\actions$ is sampled as $a\sim\zeta_{\threshold,\omega}(\cdot\mid s)$. The sigmoid transforms the emulated threshold behaviour into a probabilistic policy, where the spike-grading factor $\omega$ controls the exploration-exploitation trade-off, which is essential for training in \gls{rl}. 
In essence, we exploited Theorem \ref{thm:main} to reduce the policy search space to a two-dimensional family parametrised by $\threshold$ and $\omega$, enabling efficient learning despite unknown system dynamics.\footnote{The implementation will be published upon completion of the review.}

The resulting architecture and its parameters admit a natural and explicit interpretation rather than relying on a black-box policy representation. This also yields an efficient encoding of $\states$, which would otherwise be challenging for \glspl{snn}, since the size of $\states$ scales with $\maxaoi$.

\subsection{Benchmark ANN-Architecture}
\label{subsec:ann}
As a benchmark for the proposed solution, we use a structurally related \gls{ann}-based policy $\alpha_{w,b}\colon\states\to\Delta\actions$. Here, the probability of making a transmission attempt is given by
\begin{equation*}
    \alpha_{w,b}(1\mid s):=\sigma(w s+b)\in(0,1),
\end{equation*}
where $w,b\in\reals$ are the scalar weight and bias of the single artificial neuron, which replaces the spiking mechanism of $\zeta_{\threshold,\omega}$. As for $\zeta_{\threshold,\omega}$, the observed $\apaoi{}$-value $s$ is used as the input activation.
Consequently, $\alpha_{w,b}(0\mid s)=1-\alpha_{w,b}(1\mid s)$. The transmission decisions $a\in\actions$ are sampled as $a\sim\alpha_{w,b}(\cdot\mid s)$. Similarly to the \gls{snn}, we use $\sigma$ as the activation function, which yields a probabilistic policy suited for \gls{rl}.
Here, the magnitude of the weight $w$ controls the steepness of the sigmoid and thus the degree of randomness in the resulting policy.
In contrast to our proposed \gls{snn}, the threshold is not represented explicitly here. Instead, the quantity $-\frac{b}{w}$ can be interpreted as the effective threshold of the policy.

\section{Inference energy consumption comparison}
\label{sec:energy}
In this section, we compare the energy costs during inference of the proposed \gls{snn} solution $\zeta_{\threshold,\omega}$ and the \gls{ann} benchmark $\alpha_{w,b}$. An explicit comparison depends on the underlying hardware and may differ across applications and future advancements in chip technology. Thus, we quantify energy costs using an abstract operation-level accounting model. Both policy implementations employ the same sigmoid-based randomisation and Bernoulli-sampling procedure. The energy cost of this component depends strongly on implementation details and is therefore represented by an abstract energy cost $\mathrm{E}_{\sigma,\mathcal{B}}$.
For the pre-sigmoid process, following \cite{Davidson}, we assume that the costs $E_\mathrm{A}$ of the addition operation and $E_{W}$ of writing to memory are $\epsilon$ measured in $[J]$, and the costs $E_\mathrm{M}$ of multiplication and $E_\mathrm{R}$ of reading from memory are $5\epsilon$ \cite{Davidson}.
We assume that the input $s\in\states$ does not need to be read in.

\textbf{The \gls{snn}} policy reads the parameters $\threshold$ and $\omega$ from memory. To obtain the network's output, the input $s$ is compared to $\threshold$. Then the output is $-\omega$ if $s<\threshold$ and $\omega$ otherwise. The necessary operations are one comparison and adjusting the sign of $\omega$. Since the energy model in \cite{Davidson} does not specify the energy costs of comparison and bit-flip operations, we model both operations as incurring the energy cost $\mathrm{E}_A$. In summary,
\begin{align*}
\mathrm{E}[\zeta_{\threshold,\omega}]=\underbrace{2E_\mathrm{R}}_{\text{read parameters}}+\underbrace{E_\mathrm{A}}_{\text{if comparison}}+\underbrace{E_\mathrm{A}}_{\text{sign of output}}.
\end{align*}

\textbf{The \gls{ann}} policy reads the parameters $w$ and $b$ from memory. To obtain the networks output, the weight $w$ is multiplied by the input-value $s$, and the bias $b$ is added to the result. This requires one multiplication and one addition. In summary,
\begin{align*}
\mathrm{E}[\alpha_{w,b}]=\underbrace{2E_\mathrm{R}}_{\text{read parameters}}+\underbrace{E_\mathrm{M}}_{\text{product}}+\underbrace{E_\mathrm{A}}_{\text{add bias}}.
\end{align*}
In comparison $\mathrm{E}[\zeta_{\threshold,\omega}]=12\epsilon+\mathrm{E}_{\sigma,\mathcal{B}}<16\epsilon+\mathrm{E}_{\sigma,\mathcal{B}}=\mathrm{E}[\alpha_{w,b}]$, 
i.e., $\zeta_{\threshold,\omega}$ consumes less energy than $\alpha_{w,b}$ in our operation-based comparison. Essentially, the \gls{snn} replaces the multiplication and addition required by the \gls{ann} with a comparison and a bit-flip operation. This means that, per transmission decision, the \gls{snn} has a $25\%$ reduced energy consumption in the pre-sigmoid part compared to the \gls{ann}.
\section{Numerical evaluation}
\label{sec:results}
The parameters used are summarised in Table~\ref{tab:sys}.

\begin{table}[!ht]
\begin{center}
\caption{System model parameters}
\label{tab:sys}
\scriptsize
\begin{tabular}{|c|c|c|}
    \hline
    transmission success probability & $\ptx$ & $0.9$\\
    \hline
    wake-up probability & $\pwk$ & $0.5$\\
    \hline
    transmission energy \cite{DeSombre2026} & $E^\mathrm{Tx}$ & $64.34304\,\mathrm{mJ}$\\
    \hline
    maximum \gls{aoi} value & $\maxaoi$ & $16\,\mathrm{s}$\\
    \hline
    energy-weight coefficient & $c_\mathrm{E}$ & $300\,\mathrm{J}^{-1}$\\
    \hline
    \gls{aoi}-weight coefficient & $c_\mathrm{AoI}$ & $1\,\mathrm{s}^{-1}$\\
    \hline
\end{tabular}
\end{center}
\end{table}

We compare the mean $\apaoi{}$ and energy consumption of our proposed \gls{snn} with  the following policies:
\begin{itemize}
\item \textit{Optimal}: The optimal policy $\vartheta_{\threshold^\ast}$ (cf. Theorem~\ref{thm:main}).
\item \textit{ANN}: The trained policy introduced in Section~\ref{subsec:ann}.
\item \textit{Randomly transmit}: A policy that attempts transmission independently with probability $0.5$ at wake-up events.
\item \textit{Always transmit}: A policy minimising the mean $\apaoi{}$ at the expense of maximising energy consumption.
\item \textit{Never transmit}: A policy minimising energy consumption while yielding the largest possible mean $\apaoi{}$.
\end{itemize}
The \gls{snn} and \gls{ann} are trained using the same procedure, except for individually adjusted learning rates.

The performance in Fig.~\ref{fig:aoiandenergy} is reported as the mean over $100$ independently seeded environments, each simulated for $100$ time steps. For the \textit{SNN} and \textit{ANN} policies, the reported performance is further averaged over five independently trained policy instances with different initial thresholds $\threshold_0$.

\begin{figure}[t]
    \centering
    \includegraphics[width=1\linewidth]{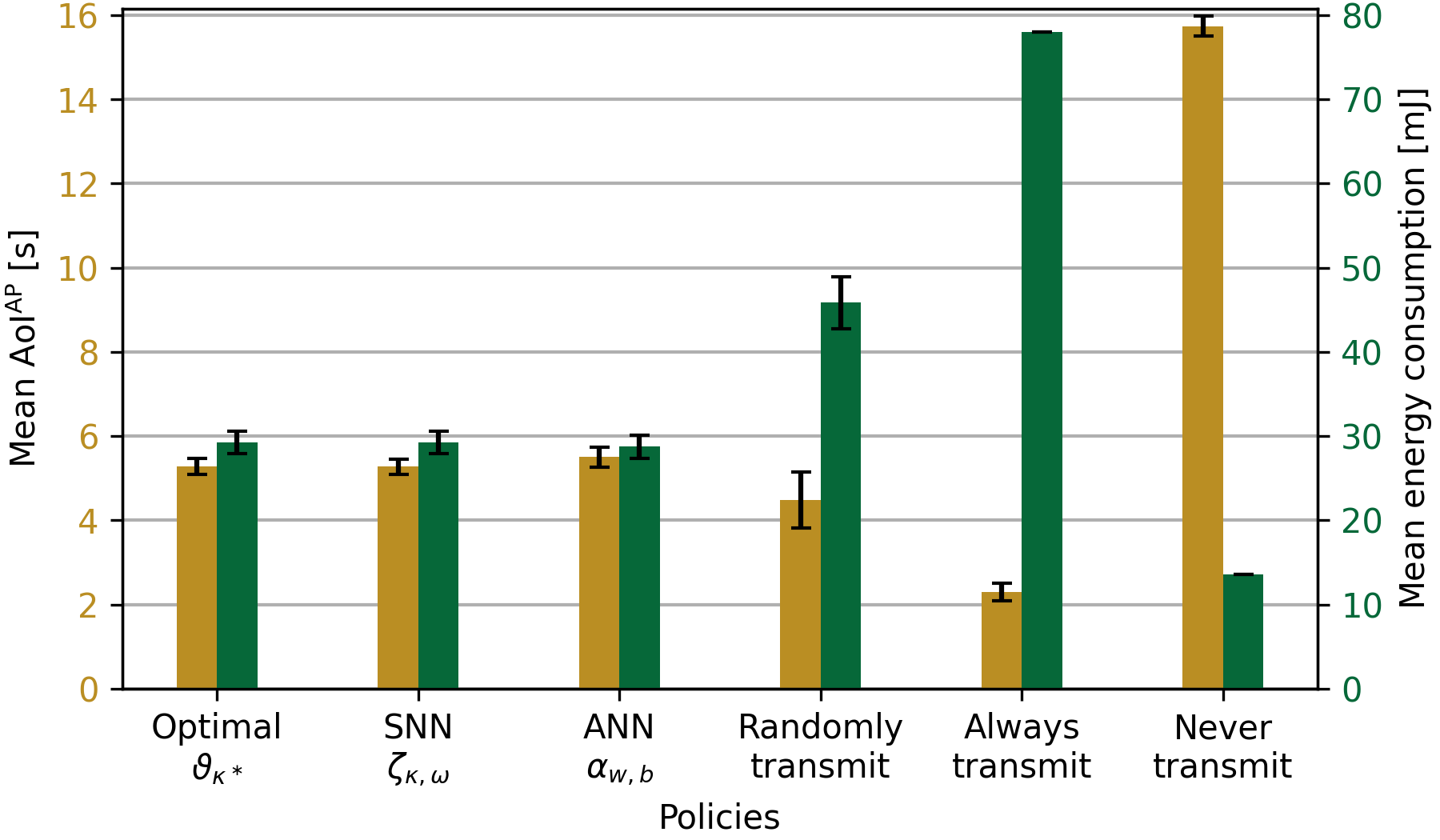}
    \caption{Energy and \gls{aoi} performance.}
    \label{fig:aoiandenergy}
\end{figure}

The results in Fig.~\ref{fig:aoiandenergy} show that the proposed \gls{snn} policy $\zeta_{\threshold,\omega}$ closely approximates the optimal policy $\vartheta_{\threshold^\ast}$ and performs on par with the \gls{ann} benchmark $\alpha_{w,b}$. The remaining performance gap results from the stochastic policy representations of the learnt policies. Compared with the \textit{Randomly transmit} policy, all three reduce energy consumption by more than $33\%$ while increasing the mean $\apaoi{}$ by less than $15\%$.
Compared with the \textit{Always transmit} policy, the optimal, \gls{snn}, and \gls{ann} policies reduce energy consumption by approximately $62\%$. Conversely, the \textit{Always transmit} policy reduces the mean $\apaoi{}$ by approximately $60\%$ relative to these policies.
The optimal policy reduces the mean $\apaoi{}$ by approximately $65\%$ relative to the \textit{Never transmit} policy. In contrast, the \textit{Never transmit} policy consumes approximately $54\%$ less energy than the optimal policy.  Note that here, the $\apaoi{}$ of the \textit{Never transmit} policy strongly depends on the value of $\maxaoi$.

The optimal threshold for the considered system parameters is $\threshold^\ast=8$. Figure~\ref{fig:trainingthreshold} shows that the proposed \gls{snn} consistently converges to the optimal threshold policy $\vartheta_{\threshold^\ast}$ despite different initial threshold parameters $\threshold_0$. 
\begin{figure}[t]
    \centering    \includegraphics[width=1\linewidth]{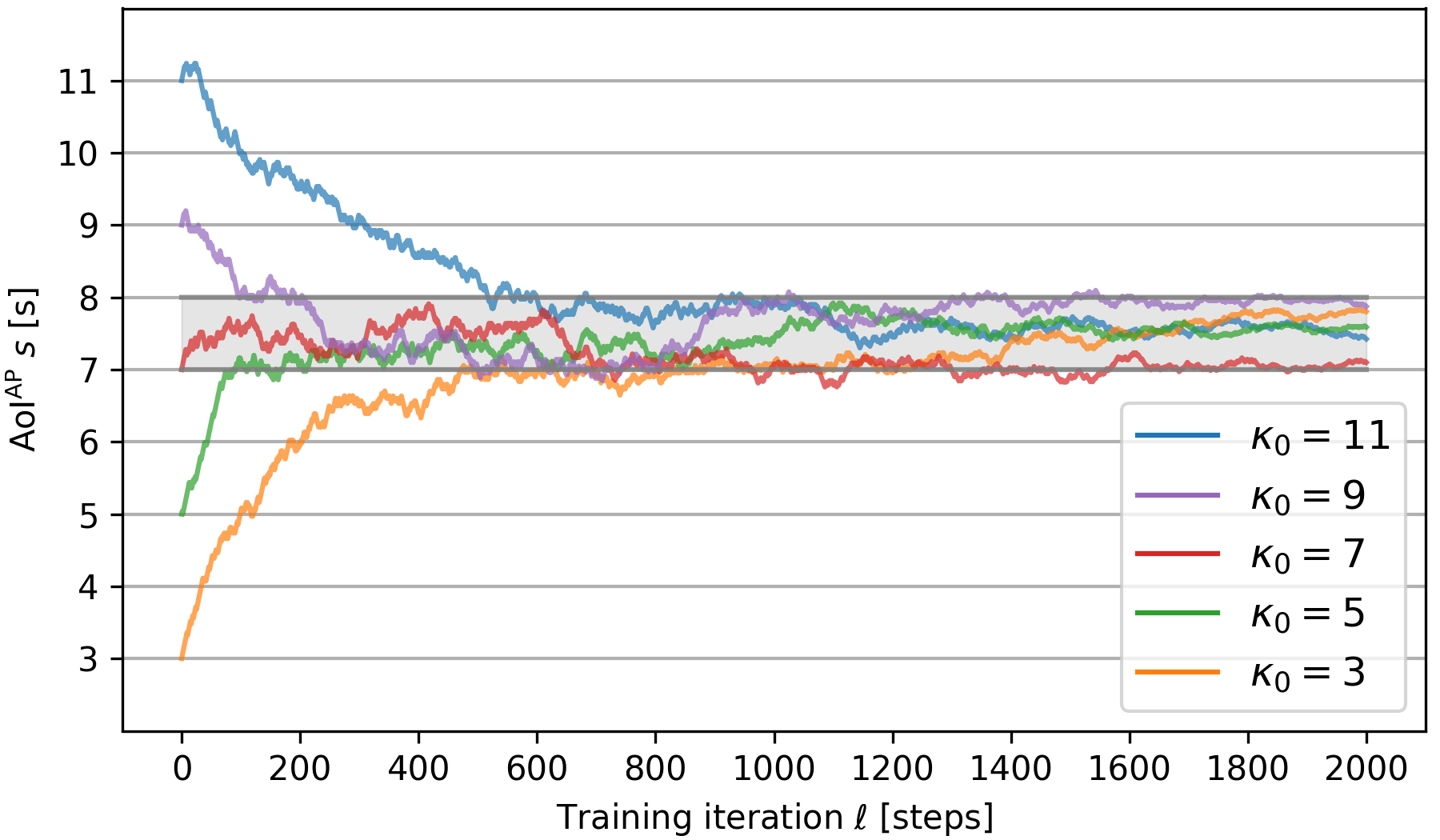}
    \caption{$\threshold_\ell$ values during training.}
    \label{fig:trainingthreshold}
\end{figure}
The grey region indicates the interval of threshold values that induce the optimal threshold policy.
After approximately $500$ training iterations, all policy instances have entered this interval and remain close to it throughout the remainder of training, despite being initialised with substantially different threshold values.

Figure~\ref{fig:thoverenvs} shows that the proposed \gls{snn} consistently learns the optimal threshold policy even when transmission energy becomes more important than information freshness.
The values of $c_\mathrm{E}$ are selected such that
$c_\mathrm{E}E^\mathrm{Tx}\approx 5,10,15,20,$ and $25$, respectively. Thus, one transmission attempt is assigned a cost equivalent to approximately $5$, $10$, $15$, $20$, or $25$ seconds of information staleness. 
As expected, increasing the relative cost of transmission leads to larger optimal thresholds, indicating that the device waits longer before attempting transmission when energy becomes more valuable relative to information freshness.
The mean values are computed over the last $1000$ training iterations, while the grey regions indicate the corresponding optimal threshold intervals.
The learnt thresholds lie within the optimal regions for all considered environments. The only exception is the second environment, in which thresholds $4$ and $5$ yield nearly identical returns. Consequently, the learnt threshold fluctuates between thresholds $4$ and $5$. The black bars indicate the minimum and maximum threshold values observed across the last $1000$ training iterations of all five trained policy instances and closely match the corresponding optimal threshold regions.

\begin{figure}[t]
    \centering
    \includegraphics[width=1\linewidth]{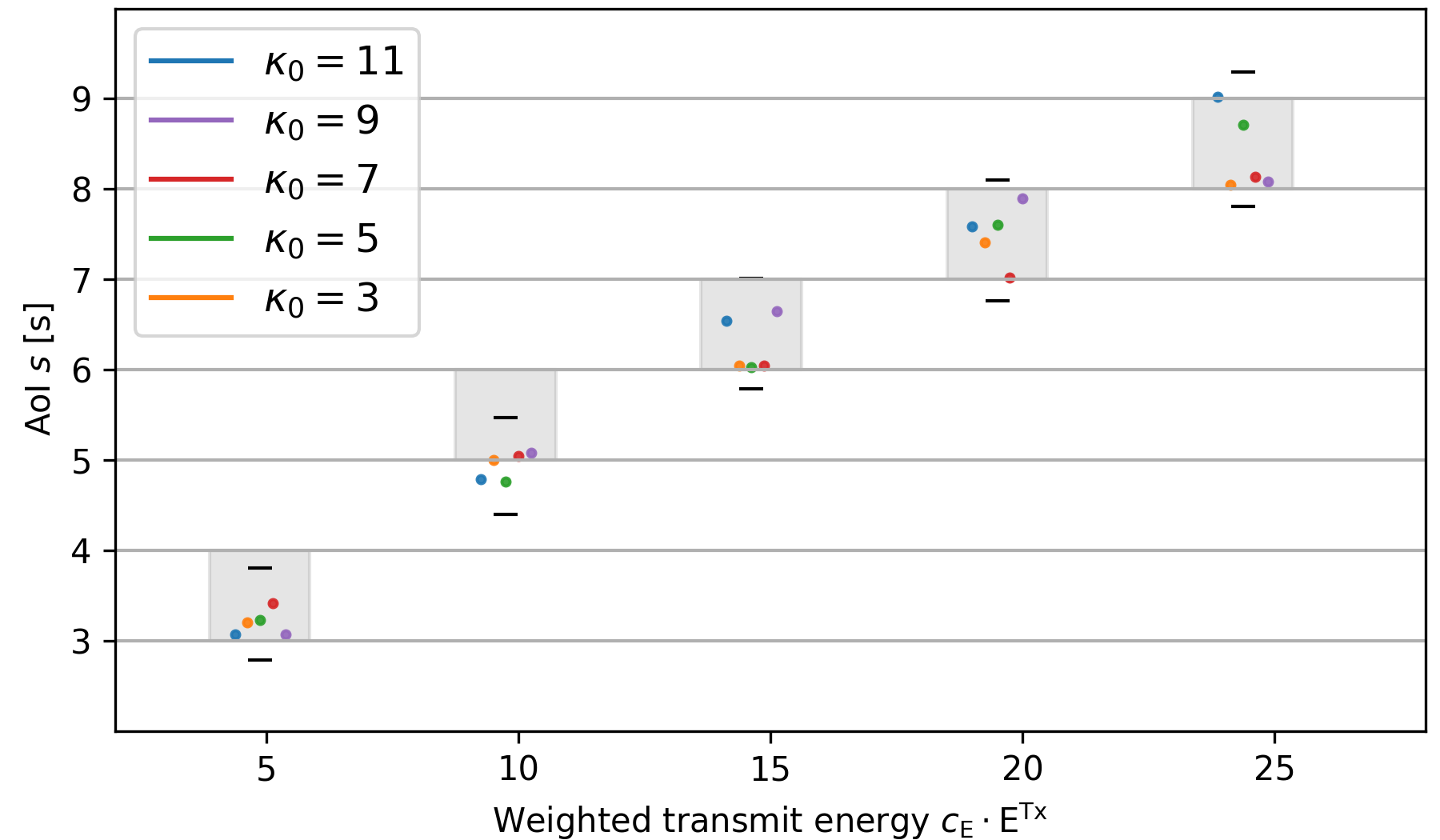}
    \caption{Mean $\threshold_\ell$ during last $1000$ training iterations.}
    \label{fig:thoverenvs}
\end{figure}

\section{Conclusion}
\label{sec:conclusion}
In this letter, we investigated transmission scheduling for event-driven status update systems in which wake-up events are governed by the monitored process. We formulated the transmission decision problem as a \gls{mdp} that jointly optimises information freshness and transmission energy and proved the existence of an optimal threshold policy. Motivated by this theoretical result, we proposed a lightweight \gls{rl} approach based on a \gls{snn} whose architecture explicitly represents threshold policies. Numerical results demonstrated that the proposed approach reliably learns near-optimal threshold policies across different operating regimes while requiring less operations than \glspl{ann}. These results suggest that \glspl{snn} provide an interpretable and energy-efficient framework for transmission scheduling in event-driven \gls{iot} systems. Future work will investigate more general wake-up processes and their impact on the structure of optimal transmission policies, as well as scenarios in which threshold optimality no longer holds.

\bibliographystyle{IEEEtran}
\bibliography{resources/IEEEabrv, resources/biblio}
\newpage
\appendices
\section{Proof of Theorem 1}
\label{sec:optpi}
In this section, we provide a proof for Theorem \ref{thm:main}. For the proof, we make use of the optimal value function 
\begin{equation*}
    \V\colon\states\to\reals;\quad s\mapsto\maximum_{\pi\colon\states\to\Delta\actions}\limit_{H\to\infty}\cexpect{\psi}{\mu(\cdot\mid\pi)}[R(\psi)\mid s_0=s]
\end{equation*}
and the optimal action-value function
\begin{equation}
    \Q\colon\states\times\actions\to\reals;\quad (s,a)\mapsto r(s,a)+\cexpect{\tilde{s}}{F(\cdot\mid s,a)}[\V(\tilde{s})]\label{QtoV}.
\end{equation}
Intuitively, $\V(s)$ captures the maximum possible cumulative reward of being in $s$, and $\Q$ captures the maximum possible cumulative reward of taking action $a$ in state $s$. Note that
\begin{equation}
    \V(s)=\maximum_{a\in\actions}\Q(s,a)\label{VtoQ}.
\end{equation}
Furthermore, we define 
\begin{equation}
     L^\ast(s,a):=\Q(s,a)-r(s,a)
\end{equation}
which is the long-term value of taking action $a$ in state $s$.
The proof broadly follows the structure given in \cite{Dongare}. Firstly, in Lemma \ref{lem:1}, we show that a policy of the form $\vartheta_\threshold$ is optimal if, with increasing $\mathrm{AoI}^\mathrm{Rx}_{w_t}$, making a transmission attempt becomes more favourable than making no transmission attempt. Then, in Lemma \ref{lem:2}, we show that this already holds when, with increasing $\mathrm{AoI}^\mathrm{Rx}_{w_t}$, the long-term value $ L^\ast(s,0)$ of making no transmission decreases. Finally, Lemma \ref{lem:3} shows that this condition is indeed satisfied in our model. 
\begin{lemma}\label{lem:1}
There is a threshold $\threshold\in\mathbb{N}$ s.t. $\vartheta_\threshold$ satisfies \eqref{eq:optProblem}, if
\begin{equation}
    D_Q\colon\states\to\reals;\quad s\mapsto\Q(s,1)-\Q(s,0)
\end{equation}
is monotonically increasing.
\begin{proof}
If $D_Q$ is monotonically increasing, then $\vartheta_\threshold$ with
\begin{equation*}
    \threshold:=
    \begin{cases}
    1,&\text{if }D_Q>0,\\
        \maxaoi+1,&\text{if }D_Q\leq0,\\
        s,&\text{ where }D_Q(s)>0\text{ and }D_Q(s-1)\leq0
    \end{cases}
\end{equation*}
satisfies \eqref{eq:optProblem}, by the definition of $\Q$.
\end{proof}
\end{lemma}
\begin{lemma}\label{lem:2}
$D_Q$ is monotonically increasing, iff \newline$ L^\ast(\cdot,0)\colon\states\to\reals$ is monotonically decreasing.
\begin{proof}
Denote $\partial L^\ast_a(s):= L^\ast(s,a)- L^\ast(s+1,a)$. Then
\begin{align*}
    \partial L^\ast_1(s)&\overset{\eqref{QtoV}}{=}\cexpect{\tilde{s}}{F(\cdot\mid s,1)}[\V(\tilde{s})]-\cexpect{\tilde{s}}{F(\cdot\mid s+1,1)}[\V(\tilde{s})]\\
    &=\sum_{\tilde{s}=1}^{\maxaoi}F(\tilde{s}\mid s,1)\V(\tilde{s})-\sum_{\tilde{s}=1}^{\maxaoi}F(\tilde{s}\mid s+1,1)\V(\tilde{s})\\
    &=\sum_{\tilde{s}=1}^{\maxaoi}\ptx\left[F(\tilde{s}\mid0,0)-F(\tilde{s}\mid0,0)\right]\V(\tilde{s})\\
    &\phantom{=}+(1-\ptx)\left[
    F(\tilde{s}\mid s,0)-F(\tilde{s}\mid s+1,0)\right]\V(\tilde{s})\\
    &=(1-\ptx)\partial L^\ast_0(s).
\end{align*}
Therefore, it follows that
\begin{align*}
    &\phantom{\iff}&&D_Q(s)\leq D_Q(s+1)\\
    &\iff&&\Q(s,1)-\Q(s,0)\leq\Q(s+1,1)-\Q(s+1,0)\\
    &\iff&&\Q(s,1)-\Q(s+1,1)\leq\Q(s,0)-\Q(s+1,0)\\
    &\iff&&\partial L^\ast_1(s)+r(s,1)-r(s+1,1)\\
    &\phantom{\iff}&&\leq\partial L^\ast_0(s)+r(s,0)-r(s+1,0)\\
    &\iff&&(1-\ptx)\partial L^\ast_0(s)+c^\mathrm{AoI}\leq \partial L^\ast_0(s)+c^\mathrm{AoI}\\
    &\iff&&\partial L^\ast_0(s)\geq0\iff L^\ast(s,0)\geq L^\ast(s+1,0),
\end{align*}
as desired.
\end{proof}
\end{lemma}
\begin{lemma}\label{lem:3}
$ L^\ast(\cdot,0)$ is monotonically decreasing.
\begin{proof}
We prove $\partial L^\ast_0(s)\geq0$ for $s<\maxaoi$ by induction on $s$.\newline
\textbf{Induction basis:} Let $s=\maxaoi-1$. We have
\begin{equation}
    \tilde{s}\sim F(\cdot\mid \maxaoi-1,0)\iff\tilde{s}\sim F(\cdot\mid \maxaoi,0)\iff\tilde{s}=\maxaoi\label{ResDistr}.
\end{equation}
Therefore, we obtain
\begin{align*}\partial L^\ast_0(\maxaoi-1)&=\cexpect{\tilde{s}}{F(\cdot\mid\maxaoi-1,0)}[\V(\tilde{s})]-\cexpect{\tilde{s}}{F(\cdot\mid \maxaoi,0)}[\V(\tilde{s})]\\
    &\overset{\eqref{ResDistr}}{=}\V(\maxaoi)-\V(\maxaoi)=0.
\end{align*}
\textbf{Induction step:} Let $s\in\states$ with $s<\maxaoi-1$ and the claim hold for all $s^\prime\in\states$ with $s^\prime>s$.
We have
\begin{align*}
F(\maxaoi-1\mid s,0)&=(1-\pwk)^{\maxaoi-1-s-1}\pwk\\
&=(1-\pwk)^{\maxaoi-(s+1)-1}\pwk+(1-\pwk)^{\maxaoi-(s+1)}\\
&\phantom{=}-(1-\pwk)^{\maxaoi-(s+1)}\left[\pwk+(1-\pwk)\right]\\
&=F(\maxaoi\mid s+1,0)-F(\maxaoi\mid s,0)
\end{align*}
and for any $\tilde{s}<\maxaoi-1$, we have
\begin{align*}
    F(\tilde{s}\mid s,0)&=\indicate[s<\tilde{s}](1-\pwk)^{\tilde{s}-s-1}\pwk\\
    &=\indicate[s+1<\tilde{s}+1](1-\pwk)^{\tilde{s}+1-(s+1)-1}\pwk\\
    &=F(\tilde{s}+1\mid s+1,0).
\end{align*}
Denote $\partial\V(s):=\V(s)-\V(s+1)$ and let
\begin{equation}
    \opti{a}\in\argmax_{a\in\actions}\Q(s+1,a)\label{amaxQ}.
\end{equation}
Then, we have
\begin{align*} 
    \partial\V(s)&\overset{\eqref{VtoQ}}{=}\maximum_{a\in\actions}\Q(s,a)-\maximum_{a\in\actions}\Q(s+1,a)\\
    &\overset{\eqref{amaxQ}}{=}\maximum_{a\in\actions}\Q(s,a)-\Q(s+1,\opti{a})\\
    &\geq\Q(s,\opti{a})-\Q(s+1,\opti{a})\\
    &\overset{\eqref{QtoV}}{=}-c_\mathrm{AoI}(s-(s+1))-c_\mathrm{E}\mathrm{E}^\mathrm{Tx}(\opti{a}-\opti{a})\\
    &\phantom{=}+\cexpect{\tilde{s}}{F(\cdot\mid s,\opti{a})}[\V(\tilde{s})]-\cexpect{\tilde{s}}{F(\cdot\mid s+1,\opti{a})}[\V(\tilde{s})]\\
    &>\partial L^\ast_{\opti{a}}(s)\geq(1-\ptx)\partial L^\ast_0(s).
\end{align*}
Thus, it follows that
\begin{align*}
    \partial L^\ast_0(s)&=\sum_{\tilde{s}=1}^{\maxaoi}F(\tilde{s}\mid s,0)\V(\tilde{s})-\sum_{\tilde{s}=1}^{\maxaoi}F(\tilde{s}\mid s+1,0)\V(\tilde{s})\\
    &=\sum_{\tilde{s}=s+1}^{\maxaoi-1}F(\tilde{s}\mid s,0)[\V(\tilde{s})-\V(\tilde{s}+1)]\\
    &>\sum_{\tilde{s}=s+1}^{\maxaoi-1}\underbrace{F(\tilde{s}\mid s,0)}_{>0}\underbrace{(1-\ptx)\partial L^\ast_0(\tilde{s})}_{\geq0\text{ by \textbf{I.H.}}}\geq0,
\end{align*}
concluding the proof.
\end{proof}
\end{lemma}
The desired result follows directly from the previous lemma.
\begin{proof}[Proof of Theorem~\ref{thm:main}]
    Lemma~\ref{lem:3} $\overset{\text{Lemma~\ref{lem:2}}}{\implies}D_Q$ is monotonically increasing$\implies$ $\vartheta_\threshold$ as defined in Lemma~\ref{lem:1} solves \eqref{eq:optProblem}.
\end{proof}
\section{Policy training}
\label{subsec:training}
The parameters of the proposed \gls{snn} are learnt using the \textsc{Reinforce} algorithm with a state-value baseline. As a policy-gradient method, \textsc{Reinforce} requires the policy to be differentiable with respect to its parameters. While the policy is differentiable with respect to the exploitation parameter $\omega$, it depends on the threshold parameter $\threshold$ through a discontinuous step function. 
Consequently, gradients with respect to $\threshold$ cannot be computed directly. Thus, we use surrogate gradients for the threshold parameter and derive the resulting gradient updates.
For optimising the parameters $\param\in\Theta$ of the policy $\pi_\Theta$, a trajectory $\psi$ of fixed length $H\in\pnats$ is sampled from interactions between the policy and the environment. For the sampled trajectory $\psi=(s_0,a_0,\ldots,s_{H-1},a_{H-1},s_H)$, by $\psi_t=(s_t,a_t,s_{t+1})$ we denote the $(t+1)$-th interaction.

In our case, the policy is trained for $2000$ optimisation steps, each of which uses mini-batches of $16$ trajectories of length $H=32$. The initial $\apaoi{}$ $s_0$ of the environment is sampled uniformly from $\states$.

Based on the sampled trajectory $\psi$, the parameter $\param\in\Theta$ of a policy $\pi_\Theta$ is updated according to the rule
\begin{align*}
    \param^\prime &= \param+\frac{\eta}{H}\sum_{t=1}^{H}\frac{\partial}{\partial\param}\mathcal{L}[\psi_t].
\end{align*}
Here, $\eta\in\reals$ is the learning rate and $\mathcal{L}[\psi_t]$ is the objective term
\begin{align*}
    \mathcal{L}[\psi_t]=\tilde{A}_t\log[\pi_\Theta(a_t\mid s_t)],
\end{align*}
where
\begin{equation*}
    \tilde{A}_t = r(s_t,a_t)+\gamma\tilde{V}(s_{t+1})-\tilde{V}(s_t)
\end{equation*}
is the \textit{TD-error}. The TD-error is computed using an estimate $\tilde{V}\colon\states\to\reals$ of the policy's value function $V^{\pi_\Theta}$. Treating the TD-error $\tilde{A}_t$ as constant with respect to the policy parameters during the policy update yields

\begin{align*}
    \frac{\partial}{\partial\param}\mathcal{L}[\psi_t]&=\tilde{A}_t\frac{\partial}{\partial\param}\log[\pi_\Theta(a_t\mid s_t)]\\
    &=\frac{\tilde{A}_t}
    {\pi_\Theta(a_t\mid s_t)}\frac{\partial}{\partial\param}\pi_\Theta(a_t\mid s_t)
\end{align*}

\section{Learning rule derivation}
For the policy $\zeta_{\threshold,\omega}$, we have $\Theta=\{\threshold,\omega\}$.
Note that
\begin{equation*}
    \zeta_{\kappa,\omega}(a_t\mid s_t)=\begin{cases}
        \sigma(\tilde{I})\quad&\text{ if }a_t=1,\\
        \sigma(-\tilde{I})\quad&\text{ if }a_t=0,
    \end{cases}
\end{equation*}
which yields
$\zeta_{\kappa,\omega}(a_t\mid s_t)=\sigma([2a_t-1]\tilde{I})$.
From this, we get
\begin{align*}
    \frac{\partial}{\partial\param}\mathcal{L}[\psi_t]&=\frac{\tilde{A}_t}
    {\zeta_{\kappa,\omega}(a_t\mid s_t)}\frac{\partial}{\partial\param}\sigma([2a_t-1]\tilde{I})\\
    &=\tilde{A}_t\frac{\sigma(\tilde{I})\sigma(-\tilde{I})}
    {\zeta_{\kappa,\omega}(a_t\mid s_t)}\frac{\partial}{\partial\param}[2a_t-1]\tilde{I}\\
    &=\tilde{A}_t[1-\zeta_{\kappa,\omega}(a_t\mid s_t)][2a_t-1]\frac{\partial}{\partial\param}\tilde{I}
\end{align*}
For the parameter $\omega$, we have
\begin{align*}
    \frac{\partial}{\partial\omega}\tilde{I}&=\frac{\partial}{\partial\omega}\omega[\vartheta_\threshold(1\mid s_t)-\vartheta_\threshold(0\mid s_t)]\\
    &=\vartheta_\threshold(1\mid s_t)-\vartheta_\threshold(0\mid s_t).
\end{align*}
Note that $[2a_t-1][\vartheta_\threshold(1\mid s_t)-\vartheta_\threshold(0\mid s_t)]=2\vartheta_\kappa(a_t\mid s_t)-1$. Hence, the update rule for the parameter $\omega$ is given by
\begin{align*}
    \frac{\partial}{\partial\omega}\mathcal{L}[\psi_t]&=\tilde{A}_t[1-\zeta_{\kappa,\omega}(a_t\mid s_t)][2\vartheta_\kappa(a_t\mid s_t)-1].
\end{align*}
Thus, the exploitation factor $\omega$ is increased if a positive TD-error $\tilde{A}_t>0$ is achieved, while the sampled action $a_t$ aligns with the decision of the threshold policy, i.e., $\vartheta_\threshold(a_t\mid s_t)=1$.

For the parameter $\threshold$, the term $\frac{\partial}{\partial\threshold}\tilde{I}$ requires determining $\frac{\partial}{\partial\threshold}\vartheta_\threshold(0\mid s_t)$ and $\frac{\partial}{\threshold}\vartheta_\threshold(1\mid s_t)$.
Since these terms are non-differentiable w.r.t $\threshold$, they cannot be used directly in the optimisation with gradient-based methods. Instead, the function
$\sigma_g(x)=\frac{1}{1+\exp(-gx)}$ with the steepness factor $g=2.5$ is employed as a differentiable surrogate function during training
\begin{align*}
    \frac{\partial}{\partial\threshold}\tilde{I}   &\approx\frac{\partial}{\partial\threshold}\omega[\sigma_g(s_t-\threshold)-\sigma_g(\threshold-s_t)].
\end{align*}

For the derivative of the surrogate gradients, we have
\begin{align*}
    \frac{\partial}{\partial\threshold}\sigma_g(\threshold-s_t)&=g\sigma_g(s_t-\threshold)\sigma_g(\threshold-s_t)=-\frac{\partial}{\partial\threshold}\sigma_g(s_t-\threshold),
\end{align*}
and thus
\begin{align*}
    \frac{\partial}{\partial\threshold}\tilde{I}
    &\approx-2\omega\frac{\partial}{\partial\threshold}\sigma_g(\threshold-s_t).
\end{align*}
Hence, the update rule for the parameter $\threshold$ is given by
\begin{align*}
    \frac{\partial}{\partial\threshold}\mathcal{L}[\psi_t]\approx\tilde{A}_t[1-\zeta_{\kappa,\omega}(a_t\mid s_t)][1-2a_t]2\omega\frac{\partial}{\partial\threshold}\sigma_g(\threshold-s_t).
\end{align*}

Consequently, a positive TD-error increases the threshold $\kappa$ when the policy chooses to transmit and decreases it when the policy chooses not to transmit.

\end{document}